\documentclass[12pt]{llncs}
\usepackage{comment}
\usepackage{graphicx}
\usepackage{indentfirst}
\usepackage{amsfonts}
\usepackage{latexsym,amsmath,epsfig,amssymb,graphics}

\usepackage[ruled,vlined,linesnumbered]{algorithm2e}

\newcommand{\keywords}[1]{\par\addvspace\baselineskip
\noindent\keywordname\enspace\ignorespaces#1}

\begin{document}

\def  \R {{\mathbb R}}
\def  \N {{\mathbb N}}
\def  \NT {{\rm NT}}

\def  \NTAlg {{\tt NonTermination}}

\title{Non-Termination Sets of Simple Linear Loops }
\author{Liyun Dai \and Bican Xia\thanks{Corresponding author}\\
    }
\institute{LMAM \& School of Mathematical Sciences, Peking University\\
 \email{dailiyun@pku.edu.cn ~ xbc@math.pku.edu.cn}
}
\date{}
\maketitle
\begin{abstract}
A simple linear loop is a simple while loop with linear assignments and linear loop guards. If a simple linear loop has only
two program variables, we give a complete algorithm for computing the set of all the inputs on which the loop
does not terminate. For the case of more program variables, we show that the non-termination set cannot
be described by Tarski formulae in general.

\keywords {Simple linear loop, termination, non-termination set, eigenvalue, Tarski formula}
\end{abstract}

\section{Introduction}\label{Sec:Intro}

Termination of programs is an important property of programs and one of the main research topics in the field of program verification.
It is well known that the following so-called ``uniform halting problem" is
undecidable in general.

\textit{
Using only a finite amount of time, determine whether a given program will always finish running or could execute forever.
}

However, there are some well known techniques for deciding termination of some special kinds of programs.
A popular technique is to use ranking functions. A ranking function for a loop maps the
values of the loop variables to a well-founded domain; further, the values of the map decrease
on each iteration. A linear ranking function is a ranking function that is a linear combination of
the loop variables and constants. Some methods for the synthesis of ranking functions and some
heuristics concerning how to automatically generate linear ranking functions
for linear programs have been proposed, for example, in Col\'{o}n and Sipma \cite{cs01}, Dams et al. \cite{DGG00}
and Podelski and Rybalchenko \cite{PR}. Podelski and Rybalchenko \cite{PR} provided an efficient and
complete synthesis method based on linear programming to construct linear ranking functions. Chen
et al. \cite{Chen} proposed a method to generate nonlinear ranking functions based on semi-algebraic
system solving. The existence of ranking function is only a sufficient condition on the
termination of a program. There are programs, which terminate, but do not have ranking functions.
Another popular technique based on well-orders, presented in Lee et al. \cite{lee}, is size-change
principle. The well-founded data can ensure that there are no
infinitely descents, which guarantees termination of programs.

For linear loops, some other methods based on calculating eigenvectors of matrices have been proposed.
Tiwari \cite{Tiw04} proved that the termination problem of a class
of linear programs (simple loops with linear loop conditions and updates) over the reals is decidable through Jordan form and
eigenvector computation. Braverman \cite{Mark06} proved that it is also
decidable over the integers. Xia et al. \cite{xia10}  considered
the termination problems of simple loops with linear updates and  polynomial
loop conditions, and proved that the termination problem of such loops over the integers
is undecidable. In \cite{xia11}, Xia et al. provided a novel symbolic decision procedure for
termination of simple linear loops, which is as efficient
as the numerical one given in \cite{Tiw04}.

A counter-example to termination is an infinite program execution. In program verification, the search for counter-examples to termination is as important as the search for proofs of termination.
In fact, these are the two folds of termination analysis of programs. Gupta et al. \cite{GuptaHMRX08} proposed a method for searching counter-examples to termination, which first enumerates lasso-shaped candidate paths for counter-examples and proves the feasibility of a given lasso by solving the existence of a recurrent set as a template-based constraint satisfaction problem. Gulwani et al. \cite{Gulwani2008} proposed a constraint-based approach to a wide class of program analyses and weakest precondition and strongest postcondition inference. The approach can be applied to generating most-general counter-examples to termination.

In this paper, we consider the set of all inputs on which a given program does not terminate. The set is called \NT\ throughout the paper. For simple linear loops, we are interested in whether the \NT\ is decidable and how to compute it if it is decidable. Similar problems was also considered in \cite{zhao11}. Our contributions in this paper are as follows. First, for homogeneous linear loops (see Section \ref{sec:pre} for the definition) with only two program variables, we give a complete algorithm for computing the \NT. For the case of more program variables, we show that the \NT\ cannot be described by Tarski formulae in general.

The rest of this paper is organized as follows.
Section \ref{sec:pre} introduces some notations and basic results on simple linear loops. Section \ref{sec:two} presents an algorithm for computing the \NT\ of homogeneous linear loops with only two program variables. The correctness of the algorithm is proved by a series of lemmas. For linear loops with more than two program variables, it is proved in Section \ref{sec:more} that the \NT\ is not a semi-algebraic set in general, i.e., it cannot be described by Tarski formulae in general. The paper is concluded in Section \ref{sec:con}.

\section{Preliminaries}\label{sec:pre}

In this paper, the domain of inputs of programs is $\R$, the field of real numbers.
A {\em simple linear loop} in general form over $\R$ can be formulated as

\[{\tt P1}: \quad {\rm while}\ \left({B\vec{x}>\vec{b}} \right )\ \left\{ {\vec{x}:=A\vec{x}+\vec{c}} \right\}
\]
where $\vec{b},\vec{c}$ are real vectors, $A_{n\times n},B_{m\times n}$ are real matrices. $B\vec{x}>\vec{b}$ is a conjunction of $m$ linear inequalities in $\vec{x}$ and $\vec{x}:=A\vec{x}+\vec{c}$ is a linear assignment on the program variables $\vec{x}$.
\begin{definition}\label{def:NT}{\rm \cite{Tiw04}}
The {\em non-termination set} of a program is the set of all inputs on which the program does not terminate. It is denoted by \NT\ in this paper.
\end{definition}

In particular,
\[{\rm NT}({\tt P1})=\{\vec{x}\in \R^n| {\tt P1}\ {\rm does\ not\ terminate\ on}\ \vec{x}\} \enspace .\]

We list some related results in \cite{Tiw04}.

\begin{proposition}\label{pro:Tiwari}{\rm\cite{Tiw04}}
For a simple linear loop {\tt P1}, the following is true.
\begin{itemize}
\item The termination of {\tt P1} is decidable.
\item If $A$ has no positive eigenvalues,  the \NT\ is empty.
\item The \NT\ is convex.
\end{itemize}
\end{proposition}

In this paper, only  the following {\em homogeneous case} is considered.
\[{\tt P2}: \quad {\rm while}\ ({B\vec{x}>0})\ \{\vec{x}:=A\vec{x}\} \enspace .
\]
Let $B_1,\ldots, B_m$ be the rows of $B$. Consider the following loops
\[L_i: \quad {\rm while}\ (B_i\vec{x}>0)\ \{\vec{x}:=A\vec{x}\} \enspace .
\]
Obviously, NT({\tt P2})=$\bigcap_{i=1}^{m} {\rm NT}(L_i).$ Therefore, without loss of generality, we assume throughout this paper that $m=1$, {\it i.e.}, there is only one inequality as the loop guard. The following is a simple example of such loops.
\[{\rm while}\ (4x_1+x_2>0)\quad
\left \{\left(\begin{array}{c} x_1 \\ x_2 \end{array} \right):=
\left(\begin{array}{cc}
-2 & 4 \\
4 & 0
\end{array}
\right)
\left(\begin{array}{c} x_1\\ x_2 \end{array} \right)\right\} \enspace .
\]
That is $B=(4,1), A=\left(
\begin{array}{cc}
-2 & 4 \\
4 & 0
\end{array}\right)\enspace .$

\section{Two-variable case}\label{sec:two}

To make things clear, we restate the problem for this two-variable case as follows.

{\it For a given homogeneous linear loop {\tt P2}
with exactly two program variables and only one inequality as the loop guard, 
compute} \NT({\tt P2}).

For simplicity, we denote the program variables by $x_1,x_2$ and use \NT\ instead of \NT({\tt P2}) in this section.
If $\vec{\alpha} $ is a non-zero point in the plane, we  denote by $\overrightarrow{\vec{\alpha}}$ a ray starting from the origin of plane and going through the point $\vec{\alpha}$.

\begin{proposition}\label{prop:2}
{\rm NT}  must be one of the  following:\\
 (1) an empty set;\\
(2) a ray starting from the origin; \\
(3) a sector between two rays starting from the origin.
\end{proposition}
\begin{proof}
We view an input $(x_1,x_2)$ as a point in the real plane with origin $O$. If there exists a point $M(x_1,x_2)\in$ NT, any point $\vec{P}$ on the ray $\overrightarrow{\vec{OM}}$  can be written as  $\vec{P}=kM=(kx_1,kx_2)$ for a positive number $k$. So $BA^n(kx_1,kx_2)^T=k^nBA^n(x_1,x_2)^T>0$ for any  $n\in \N$. That means $\vec{P}\in \NT$. Therefore, it is clear from the item 3 of
Proposition \ref{pro:Tiwari} that the conclusion is true.
\end{proof}
By the above proposition, the key point for computing the NT is to compute the ray(s) which is (are) the boundary of NT. We give the following algorithm to compute the ray(s) (and thus the NT) for {\tt P2} if the NT is not empty. The algorithm, as can be expected, is mainly based on the computation of eigenvalues and eigenvectors of $A$. The correctness of our algorithm will be proved by a series of lemmas following the algorithm.
\begin{algorithm}\label{alg:1}
\SetAlgoCaptionSeparator{. }
\caption{\NTAlg}
\DontPrintSemicolon
\KwIn{ Matrices $A_{2\times 2}$ and $B_{1\times 2}$. }
\KwOut{The NT of {\tt P2} with $A$ and $B$. }
\If {$A={\bf 0}$ or $B={\bf 0}$} {return $\emptyset$;}
Compute the eigenvalues of $A$ and denote them by $\lambda_1,\lambda_2$;\;
\If {$\lambda_1\ngtr 0\wedge \lambda_2\ngtr 0$} {return $\emptyset$;~~~~~~~~~~~~ // Proposition \ref{pro:Tiwari}}
Take $\vec{\alpha_{0}}\in \R^{2}\setminus \{0\}$ such that $B\vec{\alpha_{0}}=0$ and $BA\vec{\alpha_{0}}\geq0$;\; 
\If{$BA\vec{\alpha_{0}}=0$} {choose $\vec{\xi}$ such that $B\vec{\xi}>0$\;
\eIf{$B(A\vec{\xi})>0$} {return $\{\vec{x}|\vec{x}\in\R^{2}, B\vec{x}> 0\}$~~~~~~// Lemma \ref{le:4}}
{return $\emptyset$~~~~~~~~ // Lemma \ref{le:5}}
}
\If {$\lambda_1=0\vee \lambda_2=0$} {return $\{\vec{x}|\vec{x}\in\R^{2},B\vec{x}>0, BA\vec{x}>0\}$;~~~~ // Lemma \ref{le:6}}
Suppose $\lambda_1\geq\lambda_2$\;
\If {$\lambda_{1}\geq \lambda_{2}>0$} {choose an eigenvector $\vec{\beta_{2}}$ related to $\lambda_{2}$ such that $B\vec{\beta_{2}}\geq 0$;\;
return $\{\vec{x}|\vec{x}=k_{1}\vec{\alpha_{0}}+k_{2}\vec{\beta_{2}},k_{1}\geq 0,k_{2}>0\}$; ~~~ // Lemmas \ref{le:7} and \ref{le:8}}
\If{$\lambda_{1}>0\wedge\lambda_{2}<0$}
{\If {$\lambda_{1}\geq |\lambda_{2}|$} {let $\vec{\alpha_{-1}}=A^{-1}\vec{\alpha_{0}}$ and
 return $\{\vec{x}|\vec{x}=k_{1}\vec{\alpha_{0}}+k_{2}\vec{\alpha_{-1}},k_1>0,k_2>0\}$;} 
\If{$\lambda_{1}< |\lambda_{2}|$} {choose an eigenvector $\vec{\beta}$ related to $\lambda_{1}$ such that $B\vec{\beta}> 0$ and\; return  $\{\vec{x}|\vec{x}=k\vec{\beta},k>0\}$ ~~~ // Lemma \ref{le:10}}
}
\end{algorithm}
\begin{figure}
\begin{center}
\includegraphics[scale=0.9]{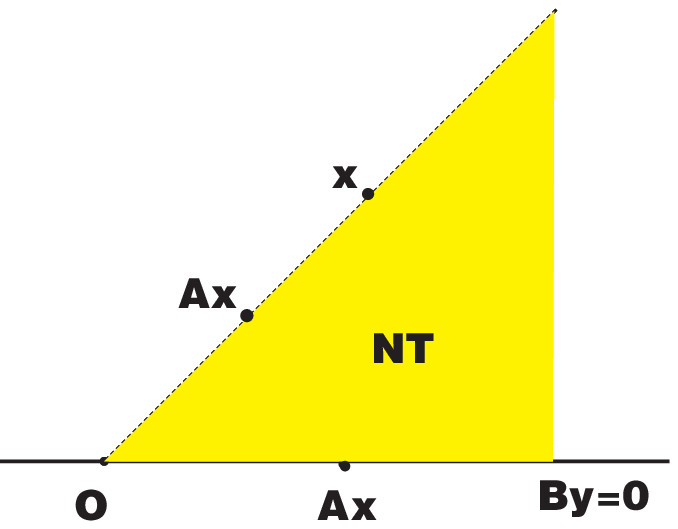}
\caption{Lemma \ref{le:1}}
\end{center}
\label{fig:le1}
\end{figure}
\begin{lemma} \label{le:1}
Suppose \NT\ is not empty and $\partial \NT $ is the boundary of \NT. If $\vec{x} \in \partial \NT $ and $B\vec{x} \neq 0$, then $A\vec{x} \in \partial \NT$.
\end{lemma}
\begin{proof}
Obviously, $B$ is a linear map from $\R^2$ to $\R$ . Because $B\vec{y}>0$ for all $\vec{y} \in \NT$, we have $B\vec{x}\geq 0$. And thus $B\vec{x}>0$ by the
assumption that $B\vec{x}\neq 0$. Hence, there exists an  open ball $o_1(\vec{x},r_1)$ such that $B\vec{y}>0$ for all
$\vec{y}\in o_1(\vec{x},r_1).$

Let $F$ be the linear map from $\R^2$ to $\R^2$ that $F(\vec{y})=A\vec{y}$ for any  $\vec{y} \in \R^2$ and hence $F$ is continuous. So for any neighborhood $o(A\vec{x},r)$ of $A\vec{x}$, there exists a positive real number $r_2$  such that $o_2(\vec{x},r_2)\subseteq o_1(\vec{x},r_1)$ and $F(o_2(\vec{x},r_2))\subseteq o(A\vec{x},r).$ Because $\vec{x} \in \partial \NT$, there exist $\vec{y},\vec{z}\in o_2(\vec{x},r_2)$ such that $\vec{y}\in {\rm NT}$ and $\vec{z}\notin {\rm NT}$. Then $A(\vec{y})$, $A(\vec{z})\in o(A\vec{x},r)$, $A\vec{(y})\in {\rm NT}$ and $ A(\vec{z})\notin {\rm NT}$. It is followed that  there are both terminating and non-terminating inputs in any neighborhood of $A\vec{x}$. Therefore, $A\vec{x} \in \partial \NT$.
\end{proof}

\begin{figure}
\begin{center}
\includegraphics[scale=0.8]{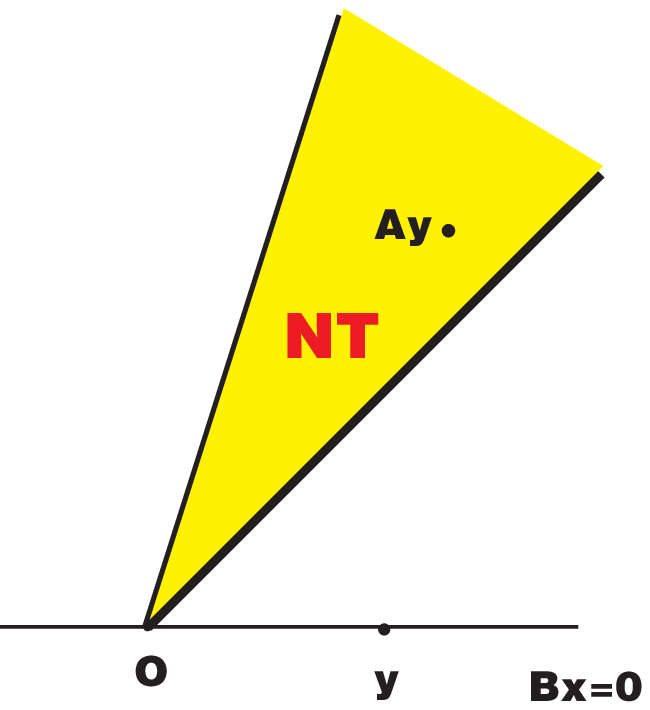}
\caption{Lemma \ref{le:2}}
\end{center}
\label{fig:le2}
\end{figure}
\begin{lemma}\label{le:2} Suppose \NT\ is neither empty nor a ray and  $\partial \NT\ \cap \{\vec{x}|B\vec{x}=0\}=\{(0,0)\}$. If $B\vec{y}=0$ and $BA\vec{y}>0$, then $A\vec{y} \in  \NT$.
\end{lemma}


\begin{proof} By Proposition \ref{prop:2}, $\partial \NT$ consists of two
rays. Let $l_1,l_2$ be the two rays. Since
neither $l_1$ nor $l_2$ is on $Bx=0$, 
$l_1$ and $l_2$  are not collinear.
 So we can choose two  points $\vec{z}\in l_1$ and $\vec{v}\in l_2$ such that
 $B\vec{z}>0$, $B\vec{v}>0$ and $\vec{y}=t_1\vec{z}+t_2\vec{v}$ for some $t_1\in
\R,t_2\in \R$. By Lemma \ref{le:1}, $A\vec{z}$ and $A\vec{v}$ must be on the
boundary of \NT, i.e., $l_1$ or $l_2$. Thus, we have at most four
possible cases as follows.
\begin{itemize}
  \item[(1)]$A\vec{z}=k_1\vec{z},A\vec{v}=k_2\vec{v},$ (i.e., $A\vec{z}\in l_1, A\vec{v}\in l_2$)
  \item[(2)]$A\vec{z}=k_1\vec{z},A\vec{v}=k_2\vec{z},$ (i.e., $A\vec{z}\in l_1, A\vec{v}\in l_1$)
   \item[(3)]$A\vec{z}=k_1\vec{v},A\vec{v}=k_2\vec{v},$ (i.e., $A\vec{z}\in l_2, A\vec{v}\in l_2$)
    \item[(4)]$A\vec{z}=k_1\vec{v},A\vec{v}=k_2\vec{z},$ (i.e., $A\vec{z}\in l_2, A\vec{v}\in l_1$)
\end{itemize}
where $k_1>0,k_2>0$.

Case (1). Because $B\vec{y}=t_1B\vec{z}+t_2B\vec{v}=0$ and
\[BA\vec{y}=BA(t_1\vec{z}+t_2\vec{v})=t_1k_1B\vec{z}+t_2k_2B\vec{v}>0,\] we have $t_1t_2<0$.
Without loss of generality,  assume that $t_1>0$ and $t_2<0$. We denote
$t_1B\vec{z}$ by $P$. Note that $P>0$ and $t_2B\vec{v}=-P$. Since
$BA\vec{y}=(k_1-k_2)P>0$, we have $k_1>k_2>0$ and
\[BA^n(A\vec{y})=k_1^{n+1}t_1B\vec{z}+k_2^{n+1}t_2B\vec{v}=k_1^{n+1}P-k_2^{n+1}P>0\]
for any $n \in \mathbb{N}$. By the definition of $\NT$, $A\vec{y} \in \NT$.

Case (2). Because $BA\vec{y}=(t_1k_1+t_2k_2)B\vec{z}>0$, we have
\[BA^n(A\vec{y})=k_1^n(t_1k_1+t_2k_2)B\vec{z}>0\] for any $n \in \mathbb{N}.$
By the definition of \NT, we have $A\vec{y} \in \NT$.

Case (3). Similarly  as Case (2), we can prove $A\vec{y} \in
\NT$.

Case (4). We shall show that this case cannot happen. Let
\[S=\{\vec{x}|\vec{x}=r_1\vec{y}+r_2A\vec{y},r_1>0,r_2>0\}\] be the sector between the two
rays $\overrightarrow{\vec{y}}$ and $\overrightarrow{\vec{Ay}}$. For any
$\vec{w} \in S$, we have $B\vec{w}=r_1B\vec{y}+r_2BA\vec{y}=r_2BA\vec{y}>0$.

Because
\[A^2\vec{y}=A(t_1k_1\vec{v}+t_2k_2\vec{z})=t_1k_1k_2\vec{z}+t_2k_1k_2\vec{v}=k_1k_2\vec{y},\]
 we have $A\vec{w}=r_1A\vec{y}+r_2A^2\vec{y}=r_1A\vec{y}+ r_2k_1k_2\vec{y}\in S$. Therefore, $\vec{w} \in \NT$ and
$S \subseteq \NT$. As $\overrightarrow{\vec{y}}$ is a boundary of
$S$ and  $B\vec{y}=0$,
$\overrightarrow{\vec{y}}$ is contained in $\partial \NT$, which
contradicts with the assumption of the lemma. So (4) cannot
happen.

In summary, $A\vec{y} \in \NT$.
\end{proof}
\begin{figure}
\begin{center}
\includegraphics[scale=0.9]{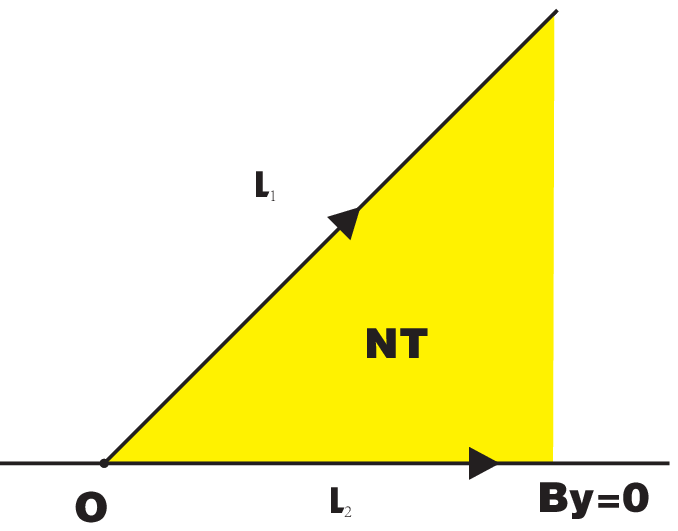}
\caption{Lemma \ref{le:3}}
\end{center}
\label{fig:le:3}
\end{figure}

\begin{lemma}\label{le:3}  If $\partial \NT$ is composed of two rays $l_1$ and $l_2$, then either $l_1$ or $l_2$ is on $B\vec{x}=0$.
\end{lemma}
\begin{proof}

Assume neither $l_1$ nor $l_2$ is   on $B\vec{x}=0$. Choose a  point $\vec{y}$ such that  $\vec{y}\neq \bf{0}$ , $ B\vec{y}=0 $ and  $BA\vec{y}\geq 0$.

Suppose $BA\vec{y}=0$. As $\NT$ is not empty,  there exists  $\vec{z} \in \NT$.  Hence $A\vec{y}$ can be rewritten  as $A\vec{y}=h_1\vec{z}+h_2\vec{y}$ for some $h_1\in \R,h_2\in \R$. As a result of $BA\vec{y}=h_1B\vec{z}+h_2B\vec{y}=h_1B\vec{z}=0$, $h_1=0$.
Note that \begin{equation}\label{eq:1}
A^n\vec{y}=h_2^n\vec{y} , BA^n\vec{y}=h_2^nB\vec{y}=0 \enspace .
\end{equation}

 According to Eq.(\ref{eq:1}) and $\vec{z} \in \NT $, we have  $BA^n(k_1\vec{z}+k_2\vec{y})=k_1BA^n\vec{z}+k_2BA^n\vec{y}=k_1BA^n\vec{z}>0$   for any   $k_1>0 ,n \in \mathbb{N}$.
Hence  $\{\vec{x}|\vec{x}=k_1\vec{z}+k_2\vec{y},k_1>0\} \subseteq \NT$. Therefore, $ \{\vec{x}|B\vec{x}=0\} =\partial \NT$, which
contradicts with the assumption.

 If $BA\vec{y}>0$,  $A\vec{y} \in \NT$ follows from  Lemma \ref{le:2}. Let $S=\{\vec{x}|k_1\vec{y}+k_2A\vec{y},k_1>0,k_2>0\}$. And we have $ BA^n\vec{z}=k_1BA^ny+k_2BA^{n+1}\vec{y}>0$  for any $n \in \mathbb{N}$,  $\vec{z} \in S$. Thus $\vec{z} \in \NT$ and  $S \subseteq \NT$. By the method of choosing $\vec{y}$, $\overrightarrow{\vec{y}} \subseteq \partial \NT$. That means $\overrightarrow{\vec{y}}$ is $l_1$ or $l_2$, which
contradicts with the assumption.
\end{proof}

\begin{lemma} \label{le:4}

Suppose $A$ has  positive eigenvalues and has an eigenvector $\vec{\alpha}$ satisfying  $B\vec{\alpha}=0$. If $\vec{\xi}$ is a vector such that $B\vec{\xi}>0$ and $BA\vec{\xi}>0$, then $\NT=\{\vec{x}|B\vec{x}>0\}$.
\end{lemma}

\begin{proof}

For any $ \vec{y}\in \{\vec{x}|B\vec{x}> 0\}$, it can be written as  $ \vec{y}=k_{1}\vec{\xi}+k_{2}\vec{\alpha}$ for some $k_1\in \R,k_2 \in \R$. As $B\vec{y}=k_{1}B\vec{\xi}+k_{2}B\vec{\alpha}=k_{1}B\vec{\xi}>0$, we have  $k_{1}>0$. Thus $BA\vec{y}=k_{1}BA\vec{\xi}+k_{2}BA\vec{\alpha}=k
_{1}BA\vec{\xi}>0$ and $ A\vec{y}\in \{\vec{x}|B\vec{x}> 0\}$. By the definition of $\NT$, we have  $\{\vec{x}|B\vec{x}> 0\}\subseteq \NT$ and hence   $\NT=\{\vec{x}|B\vec{x}> 0\}$.
\end{proof}

\begin{lemma}\label{le:5}
Suppose $A$ has positive eigenvalues and has an eigenvector $\vec{\alpha}$ satisfying $B\vec{\alpha}=0$. If there is a vector $\vec{\xi}$ such that $B\vec{\xi}>0$ and $ BA\vec{\xi}\leq 0$, then $\NT=\emptyset$.
\end{lemma}

\begin{proof}
For any $ \vec{y}\in\{\vec{x}|B\vec{x}>0\},$ it  can be written as $ \vec{y}=k_{1}\vec{\alpha}+k_{2}\vec{\xi}$ for some $k_1\in \R ,k_2\in \R$. Since $B\vec{y}=k_2B\vec{\xi}>0$, we have $k_{2}>0$. And because $BA\vec{y}=k_{2}BA\vec{\xi}\leq 0$, $\NT=\emptyset$.
\end{proof}

\begin{lemma}\label{le:6}
Suppose $A$ has a positive  eigenvalue and a zero eigenvalue. If $\vec{\gamma}$ is an eigenvector related to the positive eigenvalue such that $B\vec{\gamma }> 0$, then $\NT=\{\vec{x}|B\vec{x}>0, BA\vec{x}>0\}.$
\end{lemma}

\begin{proof}
 Let $\vec{\beta}$  be an eigenvector with respect to eigenvalue 0 and $\lambda$ be the positive eigenvalue.
Let $S$ be the set $\{\vec{x}|B\vec{x}>0,BA\vec{x}>0\}$. For any $\vec{y} \in S$, it can be written as $k_1\vec{\beta}+k_2\vec{\gamma}$ for some $k_1\in \R,k_2 \in \R$. We have $BA\vec{y}=k_2\lambda B\vec{\gamma} >0$, thus $k_2>0$.  Note that  $BA^n\vec{y}=k_2\lambda^n\vec{\gamma}>0$ for any $n \in \mathbb{N}$,
 hence $S\subseteq \NT$. Because  $\{\vec{x}|B\vec{x}\leq 0 \vee BA\vec{x}\leq 0\} \cap \NT=\emptyset $,    $\NT=\{\vec{x}|B\vec{x}>0,BA\vec{x}>0\}$.
\end{proof}

\begin{lemma}\label{le:7}
Suppose $A$ has two positive eigenvalues $\lambda_1>\lambda_2>0$ and two eigenvectors $\vec{\beta_1}$ and $\vec{\beta_2}$ related to $\lambda_1$ and $\lambda_2$, respectively,
such that $B\vec{\beta_1}>0,B\vec{\beta_2}>0$. If $\vec{\alpha}$ is a vector such that $B\vec{\alpha}=0$ and $BA\vec{\alpha}>0$, then $\NT=\{\vec{x}|\vec{x}=k_{1}\vec{\alpha}+k_{2}\vec{\beta_{2}},k_{1}\geq 0,k_{2}>0\}.$
\end{lemma}

\begin{proof}
It is easy to  know  $\vec{\beta_1},\vec{\beta_2}  \in \NT$, thus \NT\ is neither empty nor a ray. By  Lemma \ref{le:3} there is a $\overrightarrow{\vec{y}}\subseteq \partial \NT $  and $\vec{y}$ satisfies $B\vec{y}=0$. Since  for any $\vec{z} \in \partial \NT$, we have $ BA\vec{z}\geq 0$. So $BA\vec{y}\geq 0$ and hence   $\overrightarrow{\vec{\alpha}}=\overrightarrow{\vec{y}}$. In other word,  $\overrightarrow{\vec{\alpha}}$ is one ray of $\partial \NT$. Let the  other ray of  $\partial \NT$ be  $l$. As $-BA\vec{\alpha} <0$,  $\overrightarrow{\vec{-\alpha}}$ is not $l$. By Lemma \ref{le:1}, we have  $Al \in \partial \NT$. So $l$ is one of  $\overrightarrow{\vec{\beta_1}},\overrightarrow{\vec{\beta_2}}$ and $\overrightarrow{\vec{A^{-1}}\alpha}$. By directly  checking, we know  $\overrightarrow{\vec{\beta_2}}$ is $l$  and so $\NT=\{\vec{x}|\vec{x}=k_1\vec{\alpha}+k_2\vec{\beta_2},k_1\geq 0,k_2>0\}$.
\end{proof}

\begin{lemma}\label{le:8}
Assume that $A$ has one positive eigenvalue $\lambda$ with multiplicity $2$ and only one eigenvector $\vec{\beta}$
satisfying $B\vec{\beta}>0$. If $\vec{\alpha}$ is a vector such that $B\vec{\alpha}=0$ and  $BA\vec{\alpha}>0$, then $\NT=\{\vec{x}|\vec{x}=h_{1}\vec{\alpha}+h_{2}\vec{\beta},k_{1}\geq 0,k_{2}>0\}$.
\end{lemma}
\begin{proof}
By the theory of Jordan normal form in linear algebra, there exists a vector $\vec{\beta_1}$ such that $A\vec{\beta_1}=\vec{\beta}+\lambda\vec{\beta_1}$ and  $\vec{\beta}$ and $\vec{\beta_1}$ are linearly independent.

Let $\vec{\alpha_1}=A\vec{\alpha}$. We claim that
\begin{equation}\label{eq:2}
\forall n  \in \mathbb{N} . (BA^n\vec{\alpha_1}>0 \wedge \exists h_2>0. (A^n\vec{\alpha_1}=h_1\vec{\beta}+h_2\vec{\beta_1})).
\end{equation}
To prove this claim we use induction on the value of $n$.

Suppose $\vec{\alpha}=h_1\vec{\beta}+h_2\vec{\beta_1}$. If $n=0$, then $\vec{\alpha_1}=A\vec{\alpha}=(h_1\lambda+h_2)\vec{\beta}+h_2\lambda\vec{\beta_1}$. Because $ B\vec{\alpha_1}=\lambda B\vec{\alpha}+h_2 B\vec{\beta}=h_2 B\vec{\beta}>0$, we have $h_2>0$.

Now assume that the claim is true for $n-1$. Let $A^{n-1}\vec{\alpha_1}=h_1\vec{\beta}+h_2\vec{\beta_1}$ where $h_2>0$. Because $A^n\vec{\alpha_1}=A(A^{n-1}\vec{\alpha_1})=(\lambda h_1+h_2)\vec{\beta}+\lambda h_2\vec{\beta_1}$, we have $\lambda h_2>0$ and $BA^n\vec{\alpha_1}=\lambda BA^{n-1}\vec{\alpha_1}+ h_2 B\vec{\beta}>0$. So   the claim is true for any   $n \in \mathbb{N}$ and we have  $\vec{\alpha_1}\in \NT$.

Obviously, $\vec{\beta} \in \NT$ and $\vec{\beta}$ and $\vec{\alpha_1}$ are linearly independent, so \NT\ is not a ray. By Lemma \ref{le:3},   $\overrightarrow{\vec{\alpha}}\subseteq \partial \NT$.

Let the other ray of $\partial \NT $ be $l$. As $-BA\vec{\alpha}<0$, $\overrightarrow{\vec{-\alpha}}$ is not $l$. By Lemma \ref{le:1}, $Al=l $ or $ Al =\overrightarrow{\vec{\alpha}}$. So $l$ must be $\overrightarrow{\vec{\beta}}$ or $\overrightarrow{\vec{A^{-1}\alpha}}$. By  directly checking, we know  $l$ is $\overrightarrow{\vec{\beta}}$ and thus  $\NT=\{\vec{x}|\vec{x}=k_1\vec{\alpha}+k_2\vec{\beta},k_1\geq 0,k_2>0\}$.
\end{proof}

\begin{lemma}\label{le:9}
Suppose $A$ has a positive eigenvalue $\lambda_1$ and a negative eigenvalue $\lambda_2$ with $\lambda_1\geq |\lambda_2|$ and two eigenvectors $\vec{\beta_1}$ and $\vec{\beta_2}$ related to $\lambda_1$ and $\lambda_2$, respectively,
such that $B\vec{\beta_1}>0,B\vec{\beta_2}>0$. Suppose $\vec{\alpha}$ is a vector such that
$B\vec{\alpha}=0$ and $BA\vec{\alpha}>0$. Let $\vec{\alpha_{-1}}=A^{-1}\vec{\alpha}$, $\vec{\alpha_1}=A\vec{\alpha}$. Then $\NT=\{k_1\vec{\alpha}+k_2\vec{\alpha_{-1}},k_1>0,k_2>0\}$.
\end{lemma}

\begin{proof}

Let $\vec{\alpha_{-1}}=h_1\vec{\beta_1}+h_2\vec{\beta_2}$. So $\vec{\alpha}=A\vec{\alpha_{-1}}=h_1\lambda_1\vec{\beta_1}+h_2\lambda_2\vec{\beta_2}$ and $\vec{\alpha_1}=A\vec{\alpha}=h_1\lambda_1^2\vec{\beta_1}+h_2\lambda_2^2\vec{\beta_2}$.
 Because $B\vec{\alpha}=0$ and $B\vec{\alpha_1}>0$, $h_1$, $h_2$ and $A\vec{\alpha_{-1}}$ are all positive.

Note that $\vec{\alpha_1}=(-\lambda_1\lambda_2)\vec{\alpha_{-1}}+(\lambda_1+\lambda_2)\vec{\alpha}$ where $-\lambda_1\lambda_2>0$ and $\lambda_1+\lambda_2\geq 0$.
Let $S=\{\vec{x}|\vec{x}=k_1\vec{\alpha}+k_2\vec{\alpha_{-1}}$, $k_1>0$, $k_2>0\}$. Since $B\vec{y}=k_2B\vec{\alpha_{-1}}>0$  and $A\vec{y}=(k_2+k_1(\lambda_1+\lambda_2))\vec{\alpha}-k_1\lambda_1\lambda_2\vec{\alpha_{-1}}\in S$  for any $\vec{y}\in S$, we have $\NT\supseteq S$.

Let $\vec{y}=k_1\vec{\alpha}+k_2\vec{\alpha_{-1}}$. Because $B\vec{y}=k_2B\vec{\alpha_{-1}}\leq 0$ for any  $k_2\leq 0$ and  $BA\vec{y}=k_1B\vec{\alpha_1}\leq 0$ for any $k_1\leq 0$, we have $\NT=S$.
\end{proof}

\begin{lemma}\label{le:10}
Suppose A has a positive eigenvalue $\lambda_1$ and a negative eigenvalue $\lambda_2$ such that $\lambda_1<|\lambda_2|$. If there are two
eigenvectors $\vec{\beta_1}$ and $\vec{\beta_2}$ related to $\lambda_1$ and $\lambda_2$, respectively,
such that $B\vec{\beta_1}>0$ and $B\vec{\beta_2}>0$, then $\NT=\{\vec{x}|\vec{x}=k\vec{\beta_1},k>0\}$.
\end{lemma}

\begin{proof}

Consider any  $\vec{\beta}=k_1\vec{\beta_1}+k_2\vec{\beta_2} \in \R^2$.

If $k_2\neq 0$, because $A^n(k_1\vec{\beta_1}+k_2\vec{\beta_2})=k_1\lambda_1^n\vec{\beta_1}+k_2\lambda_2^n\vec{\beta_2}$ and
\[ BA^n(k_1\vec{\beta_1}+k_2\vec{\beta_2})BA^{n+1}(k_1\vec{\beta_1}+k_2\vec{\beta_2})<0\]  when $n$ is large enough,
 $k_1\vec{\beta_1}+k_2\vec{\beta_2} \notin \NT$.

If $k_2=0$, obviously, $\NT\supseteq \{\vec{x}|\vec{x}=k\vec{\beta_1},k>0\}$ and  $Bk\vec{\beta_1}\ \not \in \NT$ for any $k\leq 0$.

So $\NT= \{\vec{x}|\vec{x}=k\vec{\beta_1},k>0\}$.
\end{proof}


Now, the correctness of our algorithm \NTAlg\ can be easily obtained as follows. 

\begin{theorem}
The algorithm \NTAlg\ is correct.
\end{theorem}
\begin{proof}
First, the termination of \NTAlg\ is obvious because there are no loops and no iterations in it. Second,
it is also clear that the algorithm discusses all the cases of eigenvalues of $A$, respectively.
According to Lemmas 4-10 (each of them corresponds to a certain case in the algorithm as commented in the algorithm),
the output of the algorithm in each case is correct.
\end{proof}

\begin{example}
Compute the \NT\ of the following loop.
\begin{displaymath}
{\rm while}~ (4x_{1}+x_{2}>0)\quad
\left\{\left(
    \begin{array}{c}
      x_{1} \\
      x_{2} \\
    \end{array}
  \right)=
              \left(
                            \begin{array}{cc}
                                -2 & 4 \\
                                 4 & 0\\
                                    \end{array}
                                \right)
\left(
     \begin{array}{c}
       x_{1} \\
       x_{2} \\
     \end{array}
   \right)
   \right\}
\end{displaymath}
Herein, $B=(4,1), A=\left(
             \begin{array}{cc}
               -2 & 4 \\
               4 & 0\\
             \end{array}
           \right).
$
\end{example}
The computation of \NTAlg\ on the loop is:

Line 1. $B\neq0$ and $A\neq 0$.

Line 4. $A$ has a positive eigenvalue $-1+\sqrt{17}$.

Line 6. Let $\vec{\alpha_{0}}=(-1,4)^{T},\vec{\alpha_{1}}=A\vec{\alpha_{0}}=(18,-4)^{T}$.

Line 7. $B\vec{\alpha_{1}}=68\neq 0 $.

Line 13. The two eigenvalues of $A$ are $-1+\sqrt{17},-1-\sqrt{17}$, respectively. Neither of them is $0$.

Line 19. $A$ has two eigenvalues, of which one is positive and the other negative.

Line 20. The absolute value of the negative eigenvalue is greater than the positive eigenvalue.

Line 22. The eigenvector with respect to the positive eigenvalue is $\vec{\beta}=(1,\frac{\sqrt{17}+1}{4})^{T}$ and $B\vec{\beta}> 0$. Return  $\{\vec{x}|\vec{x}=k\vec{\beta},k>0\}$.


\section{More variables}\label{sec:more}

\begin{theorem} \label{th:1}In general, \NT\ is not a semi-algebraic set.
\end{theorem}

\begin{remark}
All Tarski formulae are in the form of conjunctions or/and disjunctions of polynomial equalities and/or inequalities, so, in other words, semi-algebraic sets are exactly the sets defined by Tarski formulae. By Theorem \ref{th:1}, we can conclude that the non-termination sets of linear loops with more than two variables cannot be defined by Tarski formulae in general.
\end{remark}

\begin{remark}
It should be noticed that all polynomial invariants are semi-algebraic sets.
\end{remark}

In order to prove the above theorem, we give an example to demonstrate its NT is not a semi-algebraic set.

\begin{proposition} Let a linear loop with three program variables be as follows.
\[
{\tt P3:}\  {\rm while}\ (x_1+2x_2+x_3\geq 0)\quad
\left\{\left(
    \begin{array}{c}
      x_1 \\
      x_2 \\
      x_3\\
    \end{array}
  \right)=
              \left(
                            \begin{array}{ccc}
                                2 & 0&0 \\
                                 0 & 3&0\\
                                 0 & 0 &5\\
                                    \end{array}
                                \right)
\left(
     \begin{array}{c}
       x_1 \\
       x_2 \\
       x_3\\
     \end{array}
   \right)
   \right\}.
\]
Then \NT{\rm ({\tt P3})} is not a semi-algebraic set.
\end{proposition}

The conclusion can be proved by using the following lemmas. For simplicity, \NT({\tt P3}) is denoted by \NT\ in this section.

\begin{lemma}\label{lem:31} 
Denote by $\tau$ the following set
\[\{9(x_1^2+x_2^2)-x_3^2<0, x_3>0\},\]
then $\tau\subseteq {\rm \NT}.$
\end{lemma}
\begin{proof}
For any $(x_1,x_2,x_3)\in\tau$, we have
$x_3>3|x_1|, x_3>3|x_2|$ and thus $x_1+2x_2+x_3>0.$ Because $A(x_1,x_2,x_3)^{T}=(2x_1,3x_2,5x_3)^{T}$ and $9(4x_1^2+9x_2^2)-25x_3^2< 0$, $A(x_1,x_2,x_3)^{T}\in \tau$. Therefore $\tau \subseteq {\rm \NT}$.
\end{proof}

\begin{lemma}\label{lem:32}$\partial \NT \subseteq \NT.$
\end{lemma}
\begin{proof}
Because the loop guard is of the form $B(x_1,x_2,x_3)^T\geq 0$, \NT\ is a closed set. So the conclusion is correct.
Furthermore, for any $(x_1,x_2,x_3)\in \partial \NT, x_1+2x_2+x_3\geq 0.$
\end{proof}

\begin{lemma}\label{lem:33}
If $(x_1,x_2,x_3)\in \NT$ and $A(x_1,x_2,x_3)^T\in \partial \NT $, then $(x_1,x_2,x_3)\in \partial \NT.$

 \end{lemma}
\begin{proof} 
Let $\vec{x}=(x_1,x_2,x_3)$.
If the conclusion is not true, there exists a  ball $o(\vec{x},r)\subseteq \NT$. Because $A\vec{x}^T\in \partial \NT$, there exists $\vec{x'}$ such that $|A\vec{x}-\vec{x'}|<r$ and $\vec{x'}$ is not in \NT.

Since $|A^{-1}\vec{x'}-\vec{x}|<|\vec{x'}-A\vec{x}|<r$, $A^{-1}\vec{x'} \in o(\vec{x},r)$.
So $A^{-1}\vec{x'} \in \NT$ and thus $\vec{x'} \in \NT$, which is a contradiction.
\end{proof}

\begin{lemma}\label{lem:34} 
$\{( \frac{1}{2^n}, -\frac{1}{3^n}, \frac{1}{5^n} )\}_{n=0}^{\infty} \subseteq \partial \NT. $
\end{lemma}
\begin{proof}
Let $\vec{p}_n=( \frac{1}{2^n}, -\frac{1}{3^n}, \frac{1}{5^n} ), n\ge 0.$
We use induction on the value of $n$.

When $n=0$, because $B\vec{p}_0=B(1,-1,1)^T=0$ and
\[ BA^k\vec{p}_0=2^k-2\times 3^k+5^k> 0\ ~~ {\rm for\ any}\ k\in \N^+,\]
we have $\vec{p}_0 \in \partial \NT.$	

Now assume that the conclusion holds for $n-1$. So, $A\vec{p}_n=\vec{p}_{n-1} \in \partial \NT \subseteq \NT.$ 
By Lemma \ref{lem:33},
$ \vec{p}_n \in  \partial \NT$.
\end{proof}

\begin{lemma}\label{lem:35}
For any non-zero polynomial $f(x_1,x_2,x_3)\in\R[x_1,x_2,x_3]$, there exists an $N$ such that $f(\frac{1}{2^n}, -\frac{1}{3^n}, \frac{1}{5^n})\neq 0$ for all $n>N$.
\end{lemma}

\begin{proof}
 Assume that the conclusion does not hold. Then there exists a subsequence $\{((\frac{1}{2})^{n_k},-(\frac{1}{3})^{n_k},(\frac{1}{5})^{n_k})\}_{k=1}^{\infty}$ such that $f$ vanishes on each point of it.

Let $f=b_1x_1^{\alpha_1}x_2^{\beta_1}x_3^{\gamma_1}+...+b_sx_1^{\alpha_s}x_2^{\beta_s}x_3^{\gamma_s}$ where $b_i\in \mathbb{R} ,b_i\neq 0 ,\alpha_i\in \mathbb{N},\beta_i \in \mathbb{N} ,\gamma_i\in\mathbb{N},$ and $(\alpha_i,\beta_i,\gamma_i)\neq (\alpha_j,\beta_j,\gamma_j)$ for $i\neq j.$

Obviously $s\geq 1$ because $f\not\equiv 0$. Let $t_i=(\frac{1}{2})^{\alpha_i}(\frac{1}{3})^{\beta_i}(\frac{1}{5})^{\gamma_i}$.

It is an obvious fact that $2^{\alpha_j}3^{\beta_j}5^{\gamma_j}\neq 2^{\alpha_i}3^{\beta_i}5^{\gamma_i}$ for $i\neq j.$
Hence $t_1,t_2,...,t_s $ are pairwise distinct. Without loss of generality, let   $t_1>t_2>...>t_s.$

For every $ j>1,$ we have $\lim\limits_{k\to\infty}{(\frac{t_j } {t_1})^{n_k} }=0$.
Thus \[\lim\limits_{k\to \infty}{|\frac{f((\frac{1}{2})^{n_k},-(\frac{1}{3})^{n_k},(\frac{1}{5})^{n_k} ) }{((\frac{1}{2})^{\alpha_1 }(\frac{1}{3})^{\beta_1 }(\frac{1}{5})^{\gamma_1})^{n_k} }|=|b_1|}\neq 0 \enspace .\]

This contradicts with $f((\frac{1}{2})^{n_k},-(\frac{1}{3})^{n_k},(\frac{1}{5})^{n_k})=0$. Therefore the conclusion follows.
\end{proof}

Using the above lemmas, we can now prove Theorem \ref{th:1}.

\begin{proof}
Denote by $S$ the sequence  $\{ (\frac{1}{2})^n,-(\frac{1}{3})^n,(\frac{1}{5})^n)\}$. By Lemma \ref{lem:34}, $S\subseteq \partial \NT.$

     Assume $\NT$ is a semi-algebraic set. Then there exist finite many polynomials $f_{i,j}\in \mathbb{R}[x_1,x_2,x_3]$ and $\triangleleft_{i,j}\in \{<, =\}$ for $i=1,...,s$ and $j=1,...,r_i$
such that
\begin{equation}
\NT=\bigcup \limits _{i=1}^s \bigcap \limits _{j=1}^{r_i} \{(x_1,x_2,x_3)\in\mathbb{R}^{3}|f_{i,j} \triangleleft_{i,j}  0\}.
\end{equation}

Because $S\subseteq \partial \NT \subseteq \{f_{i,j}=0\}_{i,j}$,   for any $x\in S$, there exists a polynomial $f_{i,j}$ such that $f_{i,j}(x)=0$.
By pigeonhole principle there exists an $f_{i,j}$ and a subsequence $S_1$ of $S$ such that  $f_{i,j}$ vanishes on $S_1$, which contradicts with Lemma \ref{lem:35}.
\end{proof}

\section{Conclusion}\label{sec:con}

In this paper, we consider whether the \NT\ of a simple linear loop is decidable and how to compute it if it is decidable. For homogeneous linear loops with only two program variables, we give a complete algorithm for computing the \NT. For the case of more program variables, we show that the \NT\ cannot be described by Tarski formulae in general.


\section*{Acknowledgements}
The work is partly supported by NNSFC 91018012 and the EXACTA project from ANR and NSFC.


\begin{thebibliography}{1}



\bibitem{Mark06}
M. Braverman: Termination of Integer Linear Programs. {\it CAV}
2006, {\it LNCS 4114}, 372--385, 2006.

\bibitem{Chen} Y. Chen, B. Xia, L. Yang, N. Zhan and C. Zhou:
Discovering Non-linear ranking functions by Solving Semi-algebraic
Systems. {\it LNCS 4711}, 34--49, 2007.



\bibitem{cs01}M. A. Col\'{o}n and H. B. Sipma: Synthesis of linear ranking functions.
{\it TACAS¡¯01}, {\it LNCS 2031}, 67--81, 2001.

\bibitem{DGG00} D. Dams, R. Gerth, and O. Grumberg: A heuristic for the automatic
generation of ranking functions. {\it Workshop on Advances in
Verification (WAVe¡¯00)}, 1--8, 2000.


\bibitem{lee}
C. S. Lee, N. D. Jones and A. M. Ben-Amram: The size-change principle for program termination.  
POPL, 81--92, 2001.


\bibitem{PR} A. Podelski and A. Rybalchenko: A complete
method for the synthesis of linear ranking functions. {\it
VMCAI}, {\it LNCS 2937}, 465--486, 2004.




\bibitem{Tiw04} A. Tiwari: Termination of Linear Programs. 
{\it CAV} 2004, {\it LNCS 3114},
70--82, 2004.



\bibitem{xia10}
B. Xia and Z. Zhang: Termination of linear programs with nonlinear constraints, {\it Journal of Symbolic Computation}, {\bf 45}: 1234--1249, 2010.

\bibitem{xia11}
 B. Xia, L. Yang, N. Zhan and Z. Zhang:
Symbolic decision procedure for termination of linear programs. {\it Formal Aspects of Computing}, {\bf 23}:171--190, 2011.

\bibitem{GuptaHMRX08}
 A. Gupta, T. Henzinger, R. Majumdar, A. Rybalchenko  and  R.-G. Xu:
 Proving non-termination, {\it POPL}, 147--158, 2008.

 \bibitem{Gulwani2008}
 S. Gulwani, S. Srivastava and R. Venkatesan: Program analysis as constraint solving, {\it POPL}, 281--292, 2008.


\bibitem{zhao11}
S. Zhao and D. Chen: Decidability Analysis on Termination Set of Loop Programs.
  	 {\it The International Conference on Computer Science and Service System(CSSS)},
3124--3127, 2011.


\end{thebibliography}
\end{document}